\documentclass[11pt]{article}
\usepackage{multibib}
\usepackage{float}
\usepackage[export]{adjustbox}
\usepackage{graphicx}
\graphicspath{ {./images/} }
\usepackage{subcaption}
\usepackage[utf8]{inputenc}
\usepackage[export]{adjustbox}
\usepackage{wrapfig}
\usepackage{authblk}
\usepackage[top=27mm, bottom=27mm, left=22mm, right=22mm]{geometry}
\usepackage{amsthm,amsmath,amssymb}
\usepackage{mathtools}
\usepackage{mathrsfs}
\usepackage{multirow}
\usepackage{microtype}
\usepackage{hyperref}
\usepackage[capitalize]{cleveref}
\usepackage{xcolor}
\usepackage{verbatim}
\usepackage{indentfirst}
\usepackage[noadjust]{cite}
\usepackage{ulem}

\usepackage{pgfplots}
\usepgfplotslibrary{groupplots}
\pgfplotsset{compat=1.18}
\usepackage{tikz}
\usepackage{float}

\newtheorem{theorem}{Theorem}[section]
\newtheorem{lemma}[theorem]{Lemma}
\newtheorem{proposition}[theorem]{Proposition}

\newtheorem{conjecture}[theorem]{Conjecture}

\newtheorem{observation}[theorem]{Observation}

\newcommand{\wt}{w_{\mathrm{H}}}
\newcommand{\supp}{\operatorname{supp}}

\title{Improved Rate-versus-Distance Upper Bounds for LDPC Codes}

\author{Chong Shangguan\thanks{C. Shangguan is with Research Center for Mathematics and Interdisciplinary Sciences, Shandong University, Qingdao 266237, China, and Frontiers Science Center for Nonlinear Expectations, Ministry of Education, Qingdao 266237, China (Email: theoreming@163.com)}, Yulin Yang\thanks{Y. Yang is with Research Center for Mathematics and Interdisciplinary Sciences, Shandong University, Qingdao 266237, China (Email: forestyoung@mail.sdu.edu.cn)}
}

\date{}

\begin{document}

\maketitle

\begin{abstract}
LDPC codes play a vital role in coding theory and practical error correction. 
A central problem in this direction is to understand their rate--distance tradeoff. 
In this paper, we develop a local-growth argument for estimating ball sizes in the coset graphs of LDPC codes. 
Rather than estimating coset balls directly, we use a local-growth analysis to bound the coset-weight generating function of linear spaces spanned by low-weight vectors.
This approach sharpens the previous ball-size estimates of Iceland and Samorodnitsky.
Combined with a general method of Friedman and Tillich that relates balls in coset graphs to sizes of error-correcting codes, 
it further improves the upper bounds on the rate of LDPC codes for a significant range of relative distances.
\end{abstract}


\section{Introduction}

\noindent A binary code $C$ of block length $n$ and minimum distance $d$ is a subset of $\{0,1\}^n$ in which any two distinct codewords differ in at least $d$ coordinates. 
Let $A(n,d)$ denote the maximum cardinality of such a code. 
The rate of $C$ is defined as $R(C)=\frac{1}{n}\log_2 |C|$, and the asymptotic rate--distance tradeoff is defined by $R(\delta)=\limsup_{n\to\infty}\frac{1}{n}\log_2 A(n,\lfloor \delta n\rfloor)$, where $\delta$ is called the relative distance. 
A central problem in coding theory is to estimate $R(\delta)$. 
Since the Plotkin bound \cite{plotkin1960binary} implies that $R(\delta)=0$ for every $\delta\ge 1/2$, throughout the paper we assume that $\delta\in(0,1/2)$ is fixed.

The best known lower bound on $R(\delta)$ is the Gilbert--Varshamov bound \cite{gilbert1952comparison,varshamov1957estimate}, namely $R(\delta)\ge 1-H(\delta)$, where $H(x)=-x\log_2x-(1-x)\log_2(1-x)$ is the binary entropy function.

The best known upper bound on $R(\delta)$ is the MRRW bound, also known as the
linear programming bound \cite{mceliece1977new}. 
The first linear programming bound is
\begin{equation*}
    R(\delta)\le R_{\mathrm{LP}}^{(1)}(\delta)=
    \textstyle H\left(\frac12-\sqrt{\delta(1-\delta)}\right),
\end{equation*}
and the second linear programming bound is
\begin{equation*}
    R(\delta)\le R_{\mathrm{LP}}^{(2)}(\delta)
    =\min_{0\le u\le 1-2\delta}
    \left(
    1+R_{\mathrm{LP}}^{(1)}(\tfrac{1-u}{2})
    -R_{\mathrm{LP}}^{(1)}(\tfrac{1-\sqrt{u^2+2\delta u+2\delta}}{2})
    \right).
\end{equation*}
Taking $u=1-2\delta$ in the above minimization gives $R_{\mathrm{LP}}^{(1)}(\delta)$, so
$R_{\mathrm{LP}}^{(2)}(\delta)\le R_{\mathrm{LP}}^{(1)}(\delta)$ for all $0<\delta<1/2$. 
However, for $0.273\le\delta<1/2$, the two bounds coincide, so the minimization over $u$ improves on $R_{\mathrm{LP}}^{(1)}(\delta)$ only for relatively small values of $\delta$.

A binary code $C\subseteq\mathbb{F}_2^n$ is linear if it is a linear subspace of $\mathbb{F}_2^n$, denoted by $C\leq\mathbb{F}_2^n$. In this paper, we focus on Low-Density Parity-Check (LDPC) codes, an important class of linear codes introduced by Gallager \cite{gallager1962low}. 
These are linear codes with sparse parity checks. 
We say that a linear code $C\le \mathbb{F}_2^n$ has parity-check density at most $w$ if it admits a parity-check matrix whose rows all have Hamming weight at most $w$; equivalently, its dual code $C^\perp$ is spanned by vectors of Hamming weight at most $w$.
LDPC codes play a central role both in coding theory and in practical error correction \cite{sipser1996expander,richardson2008modern,mackay1996near,richardson2001design,mosheiff2021low,roth2006improved,skachek2003generalized}.

A fundamental question is how the LDPC constraint affects the rate--distance tradeoff. 
For a positive integer $w$, let $R_w(\delta)$ denote the largest asymptotically achievable rate of binary linear codes with parity-check density at most $w$ and relative minimum distance at least $\delta$. 
Since $R_2(\delta)=0$ for every $\delta>0$, we restrict attention to the nontrivial regime $w\ge 3$.

Gallager \cite{gallager1962low} showed that LDPC codes asymptotically attain the Gilbert--Varshamov bound as the parity-check density grows, i.e., $\limsup_{w\to\infty} R_w(\delta)\ge 1-H(\delta)$. For more recent results on the lower bounds of $R_w(\delta)$, see \cite{litsyn2002ensembles,litsyn2003distance,chen2023improved}. 

For upper bounds, we clearly have $R_w(\delta)\leq R(\delta)\leq R_{\mathrm{LP}}^{(2)}(\delta)$. 
Burshtein, Krivelevich, Litsyn, and Miller~\cite{burshtein2002upper} showed that
\begin{equation*}
    R_w(\delta)\le 
    1-\frac{H(\delta/2)}
    {H((1-(1-\delta)^w)/2)}.
\end{equation*}
This improves the estimate $R_w(\delta)\le R_{\mathrm{LP}}^{(2)}(\delta)$ for small relative distances.

Ben-Haim and Litsyn~\cite{ben2006upper} later obtained several improved upper bounds on $R_w(\delta)$;
these bounds yield improvements for $0<\delta<1/2$ in different subranges.
One of the main tools in their work is a shortening argument, a classical technique for passing to shorter codes with controlled parameters; see, for example,~\cite[Chapter~18, Section~9]{macwilliams1977theory}. 
In the LDPC setting, their argument yields the following recursive inequality:
\begin{equation}\label{eq: recursive-shortening}
    R_w(\delta)\leq 
    \min_{0\leq t\leq 1-2\delta}
    \left(
        (1-t)R_w\left(\frac{\delta}{1-t}\right)
        +t-\frac{t}{w}
    \right).
\end{equation}
Consequently, any explicit upper bound $U_w$ on $R_w$ can be substituted into the right-hand side of \eqref{eq: recursive-shortening} in place of $R_w$.\footnote{There is no gain in iterating this step, since two successive shortenings with parameters $t$ and $s$ are equivalent to one shortening with parameter $t+(1-t)s$.} 
Taking $U_w=R_{\mathrm{LP}}^{(2)}$ gives the shortening bound~\cite[Theorem~4]{ben2006upper}
\begin{equation}\label{eq: ben2006upper}
    R_w(\delta)\leq 
    \min_{0\leq t\leq 1-2\delta}
    \left(
        (1-t)R_{\mathrm{LP}}^{(2)}\left(\frac{\delta}{1-t}\right)
        +t-\frac{t}{w}
    \right).
\end{equation}
Among the bounds obtained in~\cite{ben2006upper}, this shortening bound is the strongest one for $\delta>0.287$ when $w\ge4$ and for $\delta>0.156$ when $w=3$.

More recently, Iceland and Samorodnitsky~\cite{iceland2015coset} developed a new approach that improves $R_w(\delta)\leq R_{\mathrm{LP}}^{(1)}(\delta)$ for all $0<\delta<1/2$. 
Since $R_{\mathrm{LP}}^{(1)}(\delta)=R_{\mathrm{LP}}^{(2)}(\delta)$ for $\delta>0.273$, this also improves the bound $R_w(\delta)\leq R_{\mathrm{LP}}^{(2)}(\delta)$ for $\delta>0.273$. 
Combined with \eqref{eq: recursive-shortening}, their bounds improve upon \eqref{eq: ben2006upper} for $\delta>0.287$ when $w\geq4$, and for $\delta>0.156$ when $w=3$. 
For the detailed comparison, see~\cite[Section~5]{iceland2015coset}.
Their approach builds on Friedman and Tillich's elegant proof of the first linear programming bound for linear codes~\cite{friedman2005generalized}. 
To state their result, we first recall the definition of a coset graph.

Let $C \leq \mathbb{F}_2^n$ be a linear code. 
For a coset $A\in\mathbb F_2^n/C^\perp$, define its weight, denoted by $w_{\mathrm H}(A)$, to be the minimum Hamming weight of a vector in $A$.
The coset graph of $C$, denoted by $\mathbb T$, is the graph with vertex set $\mathbb F_2^n/C^\perp$, where two vertices $A$ and $B$ are adjacent if and only if $w_{\mathrm H}(A-B)=1$.
For $r\ge 0$, let $B_{\mathbb T}(r)=\{A\in\mathbb F_2^n/C^\perp:w_{\mathrm H}(A)\le r\}$ denote the ball of radius $r$ centered at the zero coset $\boldsymbol{0}+C^\perp$ in $\mathbb T$.

Friedman and Tillich \cite{friedman2005generalized} proved the following result.
\begin{theorem}[\cite{friedman2005generalized}]\label{thm: friedman2005generalized}
Let $C\le \mathbb{F}_2^n$ be a linear code with relative distance $\delta$, and let $\mathbb{T}$ be its coset graph. Then
    \begin{equation*}
        |C|\leq2^{o(n)}|B_{\mathbb{T}}(\rho n)|, 
    \end{equation*}
    where $\rho=\tfrac{1}{2}-\sqrt{\delta(1-\delta)}$.
\end{theorem}

Combining this with the trivial upper bound $|B_{\mathbb{T}}(r)|\leq |B(r)|$ yields the first linear programming bound, where $B(r)$ denotes the Hamming ball of radius $r$ centered at $0$.
Therefore, any improvement over the trivial upper bound on $|B_{\mathbb{T}}(r)|$ would lead to an improvement of the first linear programming bound.

Iceland and Samorodnitsky~\cite{iceland2015coset} obtained the first nontrivial estimate on $|B_{\mathbb T}(r)|$ for LDPC codes. 
For $\boldsymbol{c}\in\mathbb{F}_2^n$, let $\operatorname{supp}(\boldsymbol{c}):=\{1\le i\le n:\boldsymbol{c}_i=1\}$ denote the set of nonzero coordinates of $\boldsymbol{c}$ and let $[n]=\{1,\ldots,n\}$.
They showed that if $C\leq\mathbb{F}_2^n$ is an LDPC code of parity-check density at most $w$ and satisfies
$\bigcup_{\boldsymbol{c}\in C^\perp}\operatorname{supp}(\boldsymbol{c})=[n]$,\footnote{Otherwise, $C$ would contain a unit vector and hence have minimum distance $1$.} 
then, for every $0\leq \rho\leq 1/2$,
\begin{equation}\label{eq: iceland-ball}
    |B_{\mathbb T}(\rho n)|\leq 2^{-cn}|B(\rho n)|,
\end{equation}
where
$c=c(w,\rho)\geq \frac{\log_2 e}{8w^2}\left(\frac{\rho^w}{2}\right)^{w+1}$.
Combining this estimate with Theorem~\ref{thm: friedman2005generalized}, they obtained the following explicit upper bound:
\begin{equation}\label{eq: iceland-general}
    R_w(\delta)\leq H(\rho)-\frac{\log_2e}{8w^2}\left(\frac{\rho^{w}}{2}\right)^{w+1}, 
\end{equation}
where $\rho=\tfrac{1}{2}-\sqrt{\delta(1-\delta)}$.

For the special cases $w=3$ and $w=4$, they obtained sharper estimates for 
$|B_{\mathbb T}(\rho n)|$, which yield the following bounds on 
$R_3(\delta)$ and $R_4(\delta)$:
\begin{equation}\label{eq: iceland-w=3}
R_3(\delta)\leq 
\begin{cases}
\dfrac{2}{3},
& \delta\leq\dfrac12-\dfrac{\sqrt3}{4},\\[8pt]
\min\left\{\dfrac{2}{3},\rho+\dfrac{1}{2}H(2\rho)\right\},
& \delta> \dfrac12-\dfrac{\sqrt3}{4},
\end{cases}
\end{equation}
and
\begin{equation}\label{eq: iceland-w=4}
R_4(\delta)\leq H(\rho)+\frac{\rho}{2}\log_2 F(\rho),
\end{equation}
where $\rho=\tfrac{1}{2}-\sqrt{\delta(1-\delta)}$ and
$F(\rho)=(1-\rho)^4+4\rho(1-\rho)^3+6\rho^2(1-\rho)^2$.

Substituting \eqref{eq: iceland-w=3} into the recursive inequality \eqref{eq: recursive-shortening} improves the shortening bound \eqref{eq: ben2006upper} on $R_3(\delta)$ for $0.156<\delta<1/2$. 
Similarly, substituting \eqref{eq: iceland-w=4} for $w=4$ and \eqref{eq: iceland-general} for $w\geq5$ into the recursive inequality \eqref{eq: recursive-shortening} improves the shortening bound \eqref{eq: ben2006upper} on $R_w(\delta)$ for $0.287<\delta<1/2$.

\subsection{Our contribution}

\noindent 
In this work, we develop a local-growth argument for estimating ball sizes in coset graphs via the coset-weight generating function. 
Instead of bounding a coset ball directly, we encode the minimum Hamming weights of all cosets into a generating function and prove upper bounds for this function when the underlying linear space is generated by low-weight vectors. 
This approach allows us to control the growth of the generating function as low-weight generators are added one by one.
In this way, the problem of estimating coset balls is reduced to an extremal problem for linear spaces generated by low-weight vectors, leading to sharper LDPC rate–distance upper bounds and suggesting possible routes to further improvements (see Conjecture \ref{conjecture}).

Let $C \leq \mathbb{F}_2^n$ be a linear code. 
For $0<\lambda\leq 1$, define its coset-weight generating function by
\begin{equation*}
Q_C(\lambda)
=\sum_{A\in\mathbb{F}_2^n/C}\lambda^{w_{\mathrm H}(A)},
\end{equation*}
and set $Q_C(0)=1$ by continuity. 
Concrete examples are given in Lemmas~\ref{lem: base case} and~\ref{lem: optimal case}.

Our approach is based on the following simple observation.

\begin{observation}\label{obs: reduce}
Let $C\leq\mathbb{F}_2^n$ be a linear code, and let $\mathbb{T}$ be its coset graph. 
Then, for every $0<\lambda\leq 1$, 
\begin{equation*}
|B_{\mathbb{T}}(r)|
=\sum_{A\in\mathbb{F}_2^n/C^{\perp}:~ w_{\mathrm H}(A)\leq r}1 
\leq \sum_{A\in\mathbb{F}_2^n/C^{\perp}:~ w_{\mathrm H}(A)\leq r}\frac{\lambda^{w_{\mathrm H}(A)}}{\lambda^r} 
\leq \sum_{A\in\mathbb{F}_2^n/C^{\perp}}\frac{\lambda^{w_{\mathrm H}(A)}}{\lambda^r}
= \lambda^{-r}Q_{C^\perp}(\lambda).
\end{equation*}
\end{observation}

Combining this observation with Theorem~\ref{thm: friedman2005generalized} gives the following proposition.

\begin{proposition}\label{prop: reduce}
Let $C\leq\mathbb{F}_2^n$ be a linear code with relative distance $\delta$. Then, for every $0<\lambda\leq 1$, 
\begin{equation*}
|C|\leq 2^{o(n)}\lambda^{-\rho n}Q_{C^\perp}(\lambda), 
\end{equation*}
where $\rho=\tfrac{1}{2}-\sqrt{\delta(1-\delta)}$.
\end{proposition}

Our main result is the following theorem.

\begin{theorem}\label{thm: main-thm}
Let $D\le \mathbb{F}_2^n$ be spanned by vectors of weight at most $w$, and suppose that $\bigcup_{\boldsymbol{c}\in D}\supp(\boldsymbol{c})=[n]$. Then, for every $0\le \lambda\le 1$, 

\begin{equation*}
Q_D(\lambda)\leq\left(\frac{(1+\lambda)^w}{1+\lambda^w}\right)^{n/w}.
\end{equation*}

\noindent In particular, for $w=3$, we have

\begin{equation*}
Q_D(\lambda)\leq(1+3\lambda)^{n/3}, 
\end{equation*}

\noindent and for $w=4$, we have

\begin{equation*}
Q_D(\lambda)\leq(1+4\lambda+6\lambda^2)^{n/4}.
\end{equation*}
\end{theorem}

The bound on $Q_D(\lambda)$ for $w=3$ is sharp: when $3\mid n$, equality is attained by the code generated by the indicator vectors of disjoint $3$-sets partitioning $[n]$ (see Lemma \ref{lem: optimal case}).

Theorem~\ref{thm: main-thm} and Proposition~\ref{prop: reduce} give an upper bound on the rate of LDPC codes in terms of $\lambda$. 
For ease of presentation, we use convenient choices of $\lambda$
to obtain the following explicit bounds. For $w\geq 4$, these
choices are not generally optimal; optimization over $\lambda$
is discussed in Section~\ref{sec: proof rate-estimate}.
\begin{theorem}\label{thm: rete-estimate}
Let $0<\delta<\frac12$, and set 
$\rho=\tfrac12-\sqrt{\delta(1-\delta)}$.  Then, for every $w\geq 3$,
\begin{equation}\label{eq: general}
R_w(\delta)\le H(\rho)-\frac{1}{w}\log_2\left(1+\left(\frac{\rho}{1-\rho}\right)^w\right).
\end{equation}
In particular,
\begin{equation}\label{eq: w=3}
R_3(\delta)\leq 
\begin{cases}
\dfrac{2}{3},
& \delta\leq\dfrac12-\dfrac{\sqrt3}{4},\\[8pt]
\dfrac{2}{3}H_4(3\rho),
& \delta> \dfrac12-\dfrac{\sqrt3}{4},
\end{cases}
\end{equation}
and
\begin{equation}\label{eq: w=4}
    R_4(\delta)\leq H(\rho)+\frac{1}{4}\log_2 F(\rho),
\end{equation}
where $H_4(x)=x\log_43-x\log_4x-(1-x)\log_4(1-x)$ and
$F(\rho)=(1-\rho)^4+4\rho(1-\rho)^3+6\rho^2(1-\rho)^2$.
\end{theorem}

\paragraph{Comparison with previous results.} Lastly, we compare Theorem \ref{thm: rete-estimate} with the bounds in \cite{iceland2015coset}. 

\begin{itemize}
    \item For the general bound, using $\log_2(1+x)\geq\frac{\log_2 e}{2}x$ for $0\leq x\leq 1$ and $0<\rho<\frac{1}{2}$, we obtain
    \begin{equation*}
    \frac{\frac{1}{w}\log_2\left(1+\left(\frac{\rho}{1-\rho}\right)^w\right)}{\frac{\log_2e}{8w^2}\left(\frac{\rho^{w}}{2}\right)^{w+1}}\geq\frac{\frac{\log_2e}{2w}\left(\frac{\rho}{1-\rho}\right)^w}{\frac{\log_2e}{8w^2}\left(\frac{\rho^{w}}{2}\right)^{w+1}}=\frac{w2^{w+3}}{(1-\rho)^w\rho^{w^2}}> w2^{w^2+w+3}.
    \end{equation*}
    Hence \eqref{eq: general} strictly improves \eqref{eq: iceland-general}.

    \item For $w=3$, \eqref{eq: w=3} agrees with \eqref{eq: iceland-w=3} when $\delta\le \frac12-\frac{\sqrt3}{4}$.
    For $\frac12-\frac{\sqrt3}{4}<\delta<\frac12$, or equivalently $0<\rho<1/4$, it is easy to check that $\frac{2}{3}H_4(3\rho)<\rho+\frac{1}{2}H(2\rho)$. 
    Figure \ref{fig:w3-explicit} presents a comparison with \cite{iceland2015coset}.
    
    \item For $w=4$, since $\rho<1/2$ and $0<F(\rho)<1$, we have
    \begin{equation*}
        H(\rho)+\frac{1}{4}\log_2(F(\rho))< H(\rho)+\frac{\rho}{2}\log_2(F(\rho)).
    \end{equation*}
    Therefore, \eqref{eq: w=4} strictly improves the bound in \eqref{eq: iceland-w=4}.

    \item Moreover, substituting our bounds into the recursive inequality \eqref{eq: recursive-shortening} yields further improvements on the upper bounds for $R_w(\delta)$. 
    To the best of our knowledge, the resulting bounds are currently the best known for $\delta>0.287$ when $w\ge4$, and for $\delta>0.156$ when $w=3$. 
    Table~\ref{tab:shortening-w3} presents a numerical comparison with~\cite{iceland2015coset} after shortening for $w=3$, as the corresponding curves would be nearly indistinguishable in a plot. 
    We compute the two bounds by substituting \eqref{eq: iceland-w=3} and~\eqref{eq: w=3}, respectively, into~\eqref{eq: recursive-shortening} and minimizing over $t$ separately for each bound.
\end{itemize}

\begin{figure}[H]
\centering
\begin{tikzpicture}
\begin{axis}[
    width=0.6\textwidth,
    height=6.5cm,
    xmin=0,
    xmax=0.5,
    ymin=0,
    ymax=0.80,
    xtick={0,0.1,...,0.5},
    scaled x ticks=false,
    xticklabel style={
        /pgf/number format/fixed,
        /pgf/number format/precision=2
    },
    grid=major,
    xlabel={$\delta$},
    ylabel={Upper bound on $R_3(\delta)$},
    tick label style={font=\small},
    label style={font=\small},
    legend style={
        font=\small,
        draw=none,
        fill=none
    },
    legend pos=north east,
    legend cell align={left},
    every axis plot/.append style={line width=0.5pt},
]

\addplot+[mark=none, blue]
    table[col sep=comma, x=delta, y=rate_our_w3] {data.csv};
\addlegendentry{Our bound}

\addplot+[mark=none, red, dashed]
    table[col sep=comma, x=delta, y=rate_iceland_w3] {data.csv};
\addlegendentry{Iceland--Samorodnitsky}

\end{axis}
\end{tikzpicture}

\caption{Comparison with~\cite{iceland2015coset}
for $w=3$, before shortening.}
\label{fig:w3-explicit}
\end{figure}

\begin{table}[H]
\centering
\small
\caption{Comparison with~\cite{iceland2015coset} for $w=3$, after shortening.}
\label{tab:shortening-w3}
\begin{tabular}{cccc}
\hline
$\delta$ & Iceland--Samorodnitsky & Our bound & Improvement \\
\hline
0.16 & 0.432195 & 0.432137 & $5.72\times10^{-5}$ \\
0.20 & 0.373576 & 0.373505 & $7.15\times10^{-5}$ \\
0.30 & 0.227031 & 0.226924 & $1.07\times10^{-4}$ \\
0.40 & 0.081395 & 0.081320 & $7.49\times10^{-5}$ \\
\hline
\end{tabular}
\end{table}

\paragraph{Organization.}
The remainder of this paper is organized as follows. 
In Section~\ref{sec: preliminaries}, we derive some basic properties of the coset-weight generating function. 
In Section~\ref{sec: proof main-thm}, we prove Theorem~\ref{thm: main-thm}. 
In Section~\ref{sec: proof rate-estimate}, we prove Theorem~\ref{thm: rete-estimate}. Finally, we conclude the paper in Section~\ref{sec: rmk}.

\paragraph{Notation.} 
Throughout the paper, we use bold lowercase letters to denote vectors in $\mathbb{F}_2^n$; in particular, $\boldsymbol{0}$ and $\boldsymbol{1}$ denote the zero vector and the all-one vector in $\mathbb F_2^n$, respectively. 
For a subset $U \subseteq [n]$, let $\mathbb F_2^U$ denote the set of binary vectors indexed by $U$. 
For a vector $\boldsymbol{x} \in \mathbb F_2^n$ and a set $A \subseteq \mathbb F_2^n$, we write $\boldsymbol{x}_U = (x_i : i \in U) \in \mathbb F_2^U$ and 
$A|_U = \{\boldsymbol{a}_U : \boldsymbol{a} \in A\} \subseteq \mathbb F_2^U$.
We use $\le$ to denote inclusion of linear subspaces, and $\langle A \rangle$ to denote the linear span of a set of vectors $A$ over $\mathbb F_2$.

For disjoint subsets $U, V \subseteq [n]$ and subspaces $D \le \mathbb F_2^U$ and $E \le \mathbb F_2^V$, we identify $D$ with $D\times \{\boldsymbol{0}_V\}$ and $E$ with $\{\boldsymbol{0}_U\}\times E$ inside $\mathbb F_2^{U}\times\mathbb{F}_2^V$. 
Under this identification, we write $D\oplus E=\{(\boldsymbol{x},\boldsymbol{y}):\boldsymbol{x}\in D,\ \boldsymbol{y}\in E\}$.

\section{Properties of the coset-weight generating function}\label{sec: preliminaries}

\noindent In this section, we prove some basic properties of $Q_D(\lambda)$. 
First, we prove a monotonicity property of the coset-weight generating function.

\begin{lemma}\label{lem: monotonicity}
If $D\le E\le \mathbb{F}_2^n$, then $Q_E(\lambda)\le Q_D(\lambda)$ for every $0\le \lambda\le 1$.
\end{lemma}

\begin{proof}
It suffices to construct an injection $\pi:\mathbb{F}_2^n/E \to \mathbb{F}_2^n/D$ that preserves coset weights. For each $E$-coset $B$, choose $\boldsymbol{x}\in B$ with $w_{\mathrm H}(\boldsymbol{x})=w_{\mathrm H}(B)$ and set $\pi(B) \coloneqq\boldsymbol{x}+D \subseteq \boldsymbol{x}+E = B$. Then $w_{\mathrm H}(\boldsymbol{x}) \ge w_{\mathrm H}(\pi(B)) \ge w_{\mathrm H}(B) = w_{\mathrm H}(\boldsymbol{x})$, so $w_{\mathrm H}(\pi(B)) = w_{\mathrm H}(B)$.

To see that $\pi$ is injective, suppose $\pi(B)=\pi(B')$. Let $\boldsymbol{x}\in B$ and $\boldsymbol{x}'\in B'$ be the chosen representatives. Then $\boldsymbol{x}+D = \boldsymbol{x}'+D$, so $\boldsymbol{x}-\boldsymbol{x}' \in D\subseteq E$, and hence $B=B'$.
\end{proof}

Second, we prove a formula for computing the coset-weight generating function of a direct sum of codes with disjoint supports.

\begin{lemma}\label{lem: decomposition}
For disjoint $U,V\subseteq[n]$ and $D\le \mathbb F_2^{U}, E\leq\mathbb{F}_2^{V}$, we have $Q_{D\oplus E}(\lambda)=Q_{D}(\lambda)Q_{E}(\lambda)$.
\end{lemma}

\begin{proof}
Every coset of $D\oplus E$ has the form $A\times B$, where $A\in \mathbb F_2^U/D$ and $B\in \mathbb F_2^V/E$, and this correspondence is bijective. Since $U$ and $V$ are disjoint, $w_{\mathrm H}(A\times B)=w_{\mathrm H}(A)+w_{\mathrm H}(B)$. Therefore
\begin{equation*}
Q_{D\oplus E}(\lambda)
=\sum_{A,B}\lambda^{w_{\mathrm H}(A\times B)}
=\left(\sum_A\lambda^{w_{\mathrm H}(A)}\right)
\left(\sum_B\lambda^{w_{\mathrm H}(B)}\right)
=Q_D(\lambda)Q_E(\lambda).
\end{equation*}
\end{proof}

Lastly, we compute the coset-weight generating functions of some simple codes.

\begin{lemma}\label{lem: base case}
Let $D_1\leq\mathbb F_2^n$ be the linear code generated by the all-one vector.
Then 
\begin{equation*}
Q_{D_1}(\lambda)=\frac12\sum_{t=0}^n\binom nt\lambda^{\min\{t,~n-t\}}.
\end{equation*}
\end{lemma}

\begin{proof}
Observe that summing over $\boldsymbol{x}\in\mathbb{F}_2^n$ counts each $D_1$-coset $\boldsymbol{x}+\langle\boldsymbol{1}\rangle = \{\boldsymbol{x},\boldsymbol{x}+\boldsymbol{1}\}$ exactly twice, since both $\boldsymbol{x}$ and $\boldsymbol{x}+\boldsymbol{1}$ lie in the same coset. Hence,
\begin{equation*}
Q_{D_1}(\lambda)
=\frac12\sum_{\boldsymbol{x}\in\mathbb{F}_2^n}
\lambda^{\wt(\boldsymbol{x}+\langle\boldsymbol{1}\rangle)}=
\frac12\sum_{\boldsymbol{x}\in\mathbb{F}_2^n}
\lambda^{\min\{w_{\mathrm H}(\boldsymbol{x}),~w_{\mathrm H}(\boldsymbol{x}+\boldsymbol{1})\}}=\frac12\sum_{t=0}^n\binom nt\lambda^{\min\{t,~n-t\}}.
\end{equation*}
\end{proof}

\begin{lemma}\label{lem: optimal case}
    Suppose that $w\mid n$. Let $\boldsymbol{b}_1,\ldots,\boldsymbol{b}_{n/w}\in\mathbb{F}_2^n$ be vectors with pairwise disjoint supports, each of size $w$, such that $\bigcup_{1\le i\le n/w}\supp(\boldsymbol{b}_i)=[n]$. Let $D_{n/w}=\langle \boldsymbol{b}_1,\ldots,\boldsymbol{b}_{n/w} \rangle$. Then 
    \begin{equation*}
        Q_{D_{n/w}}(\lambda)=\left(\frac12\sum_{t=0}^w\binom wt\lambda^{\min\{t,~w-t\}}\right)^{n/w}.
    \end{equation*}
\end{lemma}

\begin{proof}
    Since the supports of $\boldsymbol{b}_1,\ldots,\boldsymbol{b}_{n/w}$ are pairwise disjoint and partition $[n]$, the code $D_{n/w}$ is the direct sum of $n/w$ copies of the one-dimensional code generated by the all-one vector in $\mathbb F_2^w$. The result follows from Lemmas \ref{lem: decomposition} and \ref{lem: base case}.
\end{proof}

\section{Proof of Theorem~\ref{thm: main-thm}}\label{sec: proof main-thm}

\paragraph{Proof overview.} 
Suppose that $D$ is a code spanned by low-weight vectors. 
We reconstruct the code $D$ (or a subcode of $D$) by adding its low-weight generators one by one.
At each step, we add one of the remaining low-weight generators, which introduces $v$ new coordinates. 
Lemma~\ref{lem: extension} bounds the one-step growth of the coset-weight generating function by the maximum of a local factor $\varPhi_{v,\Delta}(\lambda)$ for $0 \le \Delta \le w-v$.
Hence, the proof reduces to upper-bounding the local factor $\varPhi_{v,\Delta}(\lambda)$ for $v+\Delta \le w$. 
We present a general bound in Lemma~\ref{lem: general factor}, and give sharper estimates for $w=3$ and $w=4$ in Lemmas~\ref{lem: w=3 factor} and~\ref{lem: w=4 factor}.
Finally, by multiplying the one-step growth over all generators, we obtain Theorem~\ref{thm: main-thm}.

\paragraph{Formal proof.} 
For nonnegative integers $v$ and $\Delta$, define the local factor

\begin{equation*}
    \varPhi_{v,\Delta}(\lambda)=\frac{\sum_{t=0}^v\binom{v}{t}\lambda^{\min\{t,~\Delta+v-t\}}}{1+\lambda^{\Delta}}.
\end{equation*}

The following lemma is the main technical ingredient of the proof. 
It shows that the above factor controls the one-step growth of $Q_D(\lambda)$.

\begin{lemma}\label{lem: extension}
Let $U \subseteq [n]$ and $D \leq \mathbb{F}_2^U$, and let $\boldsymbol{b} \in \mathbb{F}_2^n$ satisfy $w_{\mathrm H}(\boldsymbol{b}) \leq w$ and $\supp(\boldsymbol{b}) \not\subseteq U$. Define $V = \supp(\boldsymbol{b}) \setminus U$, and set
$E= \langle D \times \{\boldsymbol{0}_V\},~ (\boldsymbol{b}_U,\boldsymbol{b}_V)\rangle \leq \mathbb{F}_2^U \times \mathbb{F}_2^V$. Here $(\boldsymbol{b}_U,\boldsymbol{b}_V)$ is the restriction of $\boldsymbol{b}$ to $U\cup V$.
Then, for every $0 < \lambda \leq 1$,
\begin{equation*}
    \frac{Q_E(\lambda)}{Q_D(\lambda)}\leq\max_{0\leq\Delta\leq w-|V|}\varPhi_{|V|,\Delta}(\lambda).
\end{equation*}
\end{lemma}

\begin{proof}
Observe that $\boldsymbol{b}_V=\boldsymbol{1}_V$. 
Choose a minimum-weight vector $\boldsymbol{z}\in \boldsymbol{b}_U+D$. 
Since $\boldsymbol{z}-\boldsymbol{b}_U\in D$, we have $(\boldsymbol{z},\boldsymbol{1}_V)-(\boldsymbol{b}_U,\boldsymbol{b}_V)\in D\times\{\boldsymbol{0}_V\}$. 
Hence replacing $(\boldsymbol{b}_U,\boldsymbol{b}_V)$ by $(\boldsymbol{z},\boldsymbol{1}_V)$ does not change $E$, so
$E=\langle D\times\{\boldsymbol{0}_V\},(\boldsymbol{z}, \boldsymbol{1}_V)\rangle$. 
By the minimality of $\boldsymbol{z}$, we have $w_{\mathrm H}(\boldsymbol{z})\le w_{\mathrm H}(\boldsymbol{b}_U)\le w-|V|$.

If $\boldsymbol{z}=\boldsymbol{0}$, then $E=D\oplus\langle\boldsymbol{1}_V\rangle$. 
By Lemmas~\ref{lem: decomposition} and~\ref{lem: base case}, we obtain $\frac{Q_E(\lambda)}{Q_D(\lambda)}=Q_{\langle\boldsymbol{1}_V\rangle}(\lambda)=\varPhi_{|V|,0}(\lambda)$.

If $\boldsymbol{z}\neq \boldsymbol{0}$, then by minimality we have $\boldsymbol{z}\notin D$. 
Translation by $\boldsymbol{z}$ partitions the $D$-cosets into pairs $\{A,A+\boldsymbol{z}\}$. 
From each pair, choose a coset of minimum weight, breaking ties arbitrarily, and let $\mathcal L_{\boldsymbol{z}}$ denote the collection of chosen cosets, i.e.,
\begin{equation*}
    \mathbb{F}_2^U/D=\bigcup_{A\in\mathcal L_{\boldsymbol{z}}}\{A,A+\boldsymbol{z}\}.
\end{equation*}
For $A\in\mathcal L_{\boldsymbol{z}}$, define 
$\Delta_A=w_{\mathrm H}(A+\boldsymbol{z})-w_{\mathrm H}(A)$. 
Choosing $\boldsymbol a\in A$ with $w_{\mathrm H}(\boldsymbol a)=w_{\mathrm H}(A)$, the triangle inequality gives
$\Delta_A\le w_{\mathrm H}(\boldsymbol{a}+\boldsymbol z)-w_{\mathrm H}(\boldsymbol a)\le w_{\mathrm H}(\boldsymbol z)\le w-|V|$.

Note that
\begin{equation*}
Q_D(\lambda)
=
\sum_{A\in\mathcal L_{\boldsymbol z}}
(\lambda^{w_{\mathrm H}(A)}+\lambda^{\wt(A+\boldsymbol z)})
=
\sum_{A\in\mathcal L_{\boldsymbol z}}\lambda^{\wt(A)}(1+\lambda^{\Delta_A}).
\end{equation*}

We next describe the $E$-cosets. For each $(\boldsymbol{x},\boldsymbol{y})\in\mathbb{F}_2^U\times\mathbb{F}_2^{V}$, we have
\begin{equation*}
    (\boldsymbol{x},\boldsymbol{y})+E=\left((\boldsymbol{x}+D)\times\{\boldsymbol{y}\}\right)\cup\left((\boldsymbol{x}+D+\boldsymbol{z})\times\{\boldsymbol{y}+\boldsymbol{1}_V\}\right).
\end{equation*}

\noindent As $(\boldsymbol{x},\boldsymbol{y})$ ranges over $\mathbb{F}_2^U\times\mathbb{F}_2^{V}$, it is straightforward to verify that the $E$-cosets are precisely
$$(\mathbb{F}_2^U\times\mathbb{F}_2^{V})/E=\{B_{A,\boldsymbol{y}}:A\in\mathcal L_{\boldsymbol z},y\in\mathbb F_2^V\},$$ 
where
\begin{equation*}
B_{A,\boldsymbol{y}} = (A\times\{\boldsymbol{y}\}) \cup ((A+\boldsymbol{z})\times\{\boldsymbol{y}+\boldsymbol{1}_V\}).
\end{equation*}
Indeed, an easy way to see this is to observe that $B_{A,\boldsymbol{y}}$ is an $E$-coset and check that $B_{A,\boldsymbol{y}}\neq B_{A',\boldsymbol{y'}}$ for $(A,\boldsymbol{y})\neq(A',\boldsymbol{y'})$, and the number of pairs $(A,\boldsymbol{y})$ with $A\in\mathcal L_{\boldsymbol z}$ and $\boldsymbol{y}\in\mathbb F_2^V$ equals $\frac12 \cdot |\mathbb F_2^U/D|\cdot 2^{|V|}$, which is exactly the number of $E$-cosets.

Therefore,

\begin{equation*}
\begin{split}
    Q_E(\lambda)
    &=\sum_{A\in\mathcal L_{\boldsymbol z}}\sum_{\boldsymbol{y}\in\mathbb F_2^V}\lambda^{\wt(B_{A,\boldsymbol{y}})}\\
    &=\sum_{A\in\mathcal L_{\boldsymbol z}}\sum_{\boldsymbol{y}\in\mathbb F_2^V}\lambda^{\min\{\wt(A)+\wt(\boldsymbol{y}),~\wt(A+\boldsymbol{z})+|V|-\wt(\boldsymbol{y})\}}\\
    &=\sum_{A\in\mathcal L_{\boldsymbol z}}\lambda^{\wt(A)}\sum_{t=0}^{|V|}\binom{|V|}{t}\lambda^{\min\{t,~\Delta_A+|V|-t\}}\\
    &=\sum_{A\in\mathcal L_{\boldsymbol z}}\lambda^{\wt(A)}(1+\lambda^{\Delta_A})\varPhi_{|V|,\Delta_A}(\lambda).
\end{split}
\end{equation*}

Since $0\le \Delta_A\le w-|V|$ for all $A\in\mathcal L_{\boldsymbol z}$, we conclude that
\begin{equation*}
\frac{Q_E(\lambda)}{Q_D(\lambda)}
=
\frac{\sum_{A\in\mathcal L_{\boldsymbol z}}\lambda^{\wt(A)}(1+\lambda^{\Delta_A})\varPhi_{|V|,\Delta_A}(\lambda)}
{\sum_{A\in\mathcal L_{\boldsymbol z}}\lambda^{\wt(A)}(1+\lambda^{\Delta_A})}
\le \max_{0\le \Delta\le w-|V|}\varPhi_{|V|,\Delta}(\lambda).
\end{equation*}
\end{proof}

The next lemma presents a general upper bound on the local factor $\varPhi_{v,\Delta}(\lambda)$ for $v+\Delta\le w$.

\begin{lemma}\label{lem: general factor}
Let $w$ be a positive integer and let $v,\Delta$ be nonnegative integers with $v+\Delta\le w$. Then, for every $0<\lambda\le 1$,
\begin{equation*}
    \varPhi_{v,\Delta}(\lambda)\le \left(\frac{(1+\lambda)^w}{1+\lambda^w}\right)^{v/w}.
\end{equation*}
\end{lemma}

\begin{proof}
Observe that for each $0\le t\le v$, 
\begin{equation*}
(1+\lambda^{\Delta+v})\lambda^{\min\{t,\Delta+v-t\}}
\le \lambda^t+\lambda^{\Delta+v-t}.
\end{equation*}
Multiplying by $\binom vt$ and summing over $t=0,\dots,v$ gives
\begin{equation*}
(1+\lambda^{\Delta+v})\sum_{t=0}^v \binom vt \lambda^{\min\{t,\Delta+v-t\}}
\le \sum_{t=0}^v \binom vt \lambda^t + \sum_{t=0}^v \binom vt \lambda^{\Delta+v-t}
= (1+\lambda^\Delta)(1+\lambda)^v.
\end{equation*}
Hence $\varPhi_{v,\Delta}(\lambda)\le (1+\lambda)^v/(1+\lambda^{\Delta+v})$.

Since $\Delta+v\le w$ and $0<\lambda\le1$, we have $1+\lambda^{\Delta+v}\ge 1+\lambda^w$. Also, $(1+\lambda^w)^{v/w}\le 1+\lambda^w$. Therefore,
\begin{equation*}
\varPhi_{v,\Delta}(\lambda)
\le \frac{(1+\lambda)^v}{1+\lambda^w}
\le \frac{(1+\lambda)^v}{(1+\lambda^w)^{v/w}}
= \left(\frac{(1+\lambda)^w}{1+\lambda^w}\right)^{v/w}.
\end{equation*}
\end{proof}

By a finer case analysis, one obtains better upper bounds for $\varPhi_{v,\Delta}(\lambda)$ when $w\in\{3,4\}$. The proofs are postponed to the appendix.

\begin{lemma}\label{lem: w=3 factor}
Let $v,\Delta$ be nonnegative integers with $v+\Delta\le 3$. Then, for every $0<\lambda\le 1$,
\begin{equation*}
    \varPhi_{v,\Delta}(\lambda)\le (1+3\lambda)^{v/3}.
\end{equation*}
\end{lemma}

\begin{lemma}\label{lem: w=4 factor}
Let $v,\Delta$ be nonnegative integers with $v+\Delta\le 4$. Then, for every $0<\lambda\le 1$,
\begin{equation*}
    \varPhi_{v,\Delta}(\lambda)\le (1+4\lambda+6\lambda^2)^{v/4}.
\end{equation*}
\end{lemma}

We are now ready to prove Theorem \ref{thm: main-thm}.
\begin{proof}[\textbf{Proof of Theorem~\ref{thm: main-thm}}]
The case $\lambda=0$ is immediate, since $Q_D(0)=1$. 
Assume henceforth that $0<\lambda\leq1$.
We first prove the theorem for general $w$.
Recall that $D\le \mathbb{F}_2^n$ is spanned by vectors of weight at most $w$ and satisfies $\bigcup_{\boldsymbol{c}\in D}\supp(\boldsymbol{c})=[n]$.
Choose $$\boldsymbol{b}^{(1)},\dots,\boldsymbol{b}^{(s)}\in D$$ as follows. 
First, choose $\boldsymbol{b}^{(1)}\in D$ to be any nonzero vector with $w_{\mathrm H}(\boldsymbol{b}^{(1)})\leq w$, and set $U_1 = \supp(\boldsymbol{b}^{(1)})$.
Suppose that $\boldsymbol{b}^{(1)},\dots,\boldsymbol{b}^{(i)}$ have been chosen. 
If $U_i=[n]$, stop. 
Otherwise, choose $\boldsymbol{b}^{(i+1)}\in D$ with $w_{\mathrm H}(\boldsymbol{b}^{(i+1)})\leq w$ such that $\supp(\boldsymbol{b}^{(i+1)})\nsubseteq U_i$, and set $U_{i+1}=U_i\cup\supp(\boldsymbol{b}^{(i+1)})$. 
Suppose the process stops at step $s$; then $U_s=[n]$.

For each $1\le i\le s$, let $$D_i=\langle \boldsymbol{b}^{(1)},\dots,\boldsymbol{b}^{(i)}\rangle|_{U_i}\leq \mathbb F_2^{U_i}.$$ 
Since $U_s=[n]$, we have $D_s=\langle \boldsymbol{b}^{(1)},\dots,\boldsymbol{b}^{(s)}\rangle\leq D$. 
By Lemma~\ref{lem: monotonicity}, $Q_D(\lambda)\leq Q_{D_s}(\lambda)$, so it remains to bound $Q_{D_s}(\lambda)$.

Set $V_1 = U_1$, and for $i\ge 2$ let $V_i = \supp(\boldsymbol{b}^{(i)}) \setminus U_{i-1}$. 
Clearly, $V_1,\dots,V_s$ are pairwise disjoint and $\bigcup_{i=1}^s V_i = [n]$. 
Thus, $\sum_{i=1}^s |V_i| = n$.

For $i=1$, we have $D_1=\langle \boldsymbol{1}_{V_1}\rangle\leq\mathbb{F}_2^{V_1}$. 
By Lemmas~\ref{lem: base case} and~\ref{lem: general factor},
\begin{equation}\label{eq: base-case}
    Q_{D_1}(\lambda)
    =
    \varPhi_{|V_1|,0}(\lambda)
    \leq
    \left(\frac{(1+\lambda)^w}{1+\lambda^w}\right)^{|V_1|/w}.
\end{equation}

For each $2\leq i\leq s$, identify $\mathbb F_2^{U_i}$ with 
$\mathbb F_2^{U_{i-1}}\times\mathbb F_2^{V_i}$. 
Then we have 
$$D_i=\langle D_{i-1}\times\{\boldsymbol{0}_{V_i}\},
(\boldsymbol{b}^{(i)}_{U_{i-1}},\boldsymbol{b}^{(i)}_{V_i})\rangle.$$ 
Therefore, by Lemmas~\ref{lem: extension} and~\ref{lem: general factor}, we obtain
\begin{equation}\label{eq: inductive-step}
    \frac{Q_{D_i}(\lambda)}{Q_{D_{i-1}}(\lambda)}
    \leq\max_{0\leq\Delta\leq w-|V_i|}\varPhi_{|V_i|,\Delta}(\lambda)\leq
    \left(\frac{(1+\lambda)^w}{1+\lambda^w}\right)^{|V_i|/w}.
\end{equation}

Multiplying the one-step bounds in \eqref{eq: base-case} and \eqref{eq: inductive-step} yields
\begin{equation*}
    Q_D(\lambda)\leq Q_{D_s}(\lambda)
    \leq \prod_{i=1}^s
    \left(\frac{(1+\lambda)^w}{1+\lambda^w}\right)^{|V_i|/w}
    =\left(\frac{(1+\lambda)^w}{1+\lambda^w}\right)^{\sum_{i=1}^s|V_i|/w}=
    \left(\frac{(1+\lambda)^w}{1+\lambda^w}\right)^{n/w}.
\end{equation*}

For $w=3$ and $w=4$, replacing Lemma~\ref{lem: general factor} with Lemmas~\ref{lem: w=3 factor} and \ref{lem: w=4 factor}, respectively, gives $Q_D(\lambda)\leq (1+3\lambda)^{n/3}$ (for $w=3$) and $Q_D(\lambda)\leq (1+4\lambda+6\lambda^2)^{n/4}$ (for $w=4$). We omit the details.
\end{proof}

\section{Proof of Theorem~\ref{thm: rete-estimate}}\label{sec: proof rate-estimate}

\noindent 
For each $n$, let $C\leq\mathbb{F}_2^n$ be an LDPC code of maximum cardinality with parity-check density at most $w$ and relative minimum distance at least $\delta$.
By applying Theorem~\ref{thm: main-thm} to $C^\perp$, and combining it with Proposition~\ref{prop: reduce}, we obtain
\begin{equation*}
    |C|\leq 2^{o(n)}\lambda^{-\rho n}\left(\frac{(1+\lambda)^w}{1+\lambda^w}\right)^{n/w},
\end{equation*}
where $\rho=\tfrac{1}{2}-\sqrt{\delta(1-\delta)}$. 
Since $0<\rho<1/2$, we choose $\lambda=\frac{\rho}{1-\rho}\in(0,1)$. This gives
\begin{equation*}
    |C|\leq 2^{o(n)}\left(\frac{\left(\frac{\rho}{1-\rho}\right)^{-\rho}\left(\frac{1}{1-\rho}\right)}{\left(1+\left(\frac{\rho}{1-\rho}\right)^w\right)^{1/w}}\right)^n=2^{o(n)}\left(\frac{2^{H(\rho)}}{\left(1+\left(\frac{\rho}{1-\rho}\right)^w\right)^{1/w}}\right)^n, 
\end{equation*}
and hence
\begin{equation*}
R_w(\delta)=\limsup_{n\to\infty}\frac{\log_2|C|}{n}\leq H(\rho)-\frac{1}{w}\log_2\left(1+\left(\frac{\rho}{1-\rho}\right)^w\right).    
\end{equation*}

We now prove the sharper bounds for $w=3$ and $w=4$.

\begin{itemize}
    \item For $w=3$, we have
\begin{equation*}
    |C|\leq 2^{o(n)}\lambda^{-\rho n}(1+3\lambda)^{n/3}.
\end{equation*}
If $\rho>\frac14$, we choose $\lambda=1$, which gives $|C|\leq 2^{\frac{2}{3}n+o(n)}$ and hence $R_3(\delta)\leq\frac{2}{3}$. 
If $\rho\leq\frac{1}{4}$, we choose $\lambda=\frac{\rho}{1-3\rho}$ and obtain
\begin{equation*}
R_3(\delta)=\limsup_{n\to\infty}\frac{\log_2|C|}{n}\leq -\rho\log_2\left(\frac{\rho}{1-3\rho}\right)+\frac{1}{3}\log_2\left(\frac{1}{1-3\rho}\right)=\frac{2}{3}H_4(3\rho).
\end{equation*}

\item For $w=4$, we have
\begin{equation*}
    |C|\leq 2^{o(n)}\lambda^{-\rho n}(1+4\lambda+6\lambda^2)^{n/4}.
\end{equation*}
Choosing $\lambda=\frac{\rho}{1-\rho}$, we get
\begin{equation*}
\begin{split}
R_4(\delta)
&=\limsup_{n\to\infty}\frac{\log_2|C|}{n}\\
&\leq -\rho\log_2\left(\frac{\rho}{1-\rho}\right)+\frac{1}{4}\log_2\left(1+4\left(\frac{\rho}{1-\rho}\right)+6\left(\frac{\rho}{1-\rho}\right)^2\right)\\
&= -\rho\log_2\left(\frac{\rho}{1-\rho}\right)-\log_2(1-\rho)+\frac{1}{4}\log_2\left((1-\rho)^4+4\rho(1-\rho)^3+6\rho^2(1-\rho)^2\right)\\
&=H(\rho)+\frac{1}{4}\log_2(F(\rho)).
\end{split}
\end{equation*}
\end{itemize}

\paragraph{Optimal parameter choices.}
Optimizing the general estimate above over $\lambda$ gives
\begin{equation*}
R_w(\delta)\leq
\inf_{0<\lambda\leq1}
\left\{
-\rho\log_2\lambda+\log_2(1+\lambda)
-\frac1w\log_2(1+\lambda^w)
\right\}.
\end{equation*}
The minimum occurs at $\lambda=1$ or at an interior point satisfying
\begin{equation*}
\rho=\frac{\lambda}{1+\lambda}
-\frac{\lambda^w}{1+\lambda^w}.
\end{equation*}
The special bound for $w=3$ is already optimized
in~\eqref{eq: w=3}.
For $w=4$, the special bound is minimized at
\begin{equation*}
\lambda_*=
\begin{cases}
\displaystyle
\frac{2\rho}{1-4\rho+\sqrt{1+4\rho-8\rho^2}},
& 0<\rho<4/11,\\[6pt]
1,
& 4/11\leq\rho<1/2.
\end{cases}
\end{equation*}
We use the smaller of the resulting bound and the optimized
general bound.

At $\delta=0.1$, optimization before shortening lowers the bounds
from $0.711982$ to $0.710587$ for $w=4$, and from $0.721646$
to $0.721642$ for $w=5$.
After shortening, with $t$ optimized separately for each bound,
the corresponding decreases are approximately
$1.02\times10^{-8}$ and $8.85\times10^{-15}$, respectively.

\section{Concluding Remarks}\label{sec: rmk}
\noindent We conclude this paper by mentioning a conjecture on the upper bound of $Q_D(\lambda)$.
Suppose that $w\mid n$, and recall the code $D_{n/w}$ defined in Lemma \ref{lem: optimal case}. 
We conjecture that, among all linear codes $D\le \mathbb F_2^n$ spanned by vectors of weight at most $w$ and satisfying $\bigcup_{\boldsymbol c\in D}\operatorname{supp}(\boldsymbol c)=[n]$, the quantity $Q_D(\lambda)$ is maximized by $D_{n/w}$.

\begin{conjecture}\label{conjecture}
Suppose that $w\mid n$.
Let $D\le \mathbb{F}_2^n$ be spanned by vectors of Hamming weight at most $w$, and suppose that
$\bigcup_{c\in D}\supp(c)=[n]$.
Then, for every $0\le \lambda\le 1$,
\begin{equation*}
    Q_D(\lambda)\leq \varPhi_{w,0}(\lambda)^{n/w}.
\end{equation*}
\end{conjecture}

Since $w$ is fixed and we are interested in asymptotic rate estimates, the divisibility assumption $w\mid n$ is expected to affect only lower-order terms.
Theorem~\ref{thm: main-thm} shows that the above conjecture is true for $w=3$.
For $w=4$, however, Theorem~\ref{thm: main-thm} gives
\begin{equation*}
    Q_D(\lambda)\leq (1+4\lambda+6\lambda^2)^{n/4},
\end{equation*}
whereas the conjectured bound is
\begin{equation*}
    Q_D(\lambda)\leq (1+4\lambda+3\lambda^2)^{n/4}.
\end{equation*}

Even if Conjecture~\ref{conjecture} holds, the extremal example
$D_{n/w}$ does not establish sharpness of the resulting rate bounds
at fixed positive relative distance, since its dual has minimum
distance $2$.
A natural direction is therefore to seek sharper bounds on
$Q_D(\lambda)$ when $D^\perp$ has relative minimum distance
bounded away from zero.

\section*{Acknowledgements}
\noindent We thank Zhicong Lin, Yuanting Shen, and Zihao Zhang for several discussions in the early stage of this project. We also thank Yiwei Zhang and Sihuang Hu for a helpful discussion in the late stage of this project.
The research of Chong Shangguan is supported by the National Natural Science Foundation of China under Grant Nos. 12571352 and 12231014, and the Fundamental Research Funds for the Central Universities.

\bibliographystyle{plain}
\normalem
\bibliography{ref}

\section*{Appendix}

\begin{proof}[\textbf{Proof of Lemma~\ref{lem: w=3 factor}}]
If $v=0$, then $\varPhi_{0,\Delta}(\lambda)=1/(1+\lambda^\Delta)\leq 1=(1+3\lambda)^0$. 
Assume $v\geq 1$. Since $v+\Delta\leq 3$, the possible pairs are
$(1,0),(1,1),(1,2),(2,0),(2,1),(3,0)$. 
A direct calculation gives
\begin{equation*}
\begin{array}{lll}
\varPhi_{1,0}=\varPhi_{1,1}=1, 
& \displaystyle \varPhi_{1,2}=\frac{1+\lambda}{1+\lambda^2},
& \varPhi_{2,0}=1+\lambda,\\[1.2ex]
\displaystyle \varPhi_{2,1}=\frac{1+3\lambda}{1+\lambda},
& \varPhi_{3,0}=1+3\lambda. &
\end{array}
\end{equation*}
The cases $(1,0),(1,1)$ and $(3,0)$ are immediate. 
It remains to check
\begin{equation*}
\frac{1+\lambda}{1+\lambda^2}\leq (1+3\lambda)^{1/3},\qquad
1+\lambda\leq (1+3\lambda)^{2/3},\qquad
\frac{1+3\lambda}{1+\lambda}\leq (1+3\lambda)^{2/3}.
\end{equation*}
After rearranging, it is straightforward to verify that
\begin{equation*}
\begin{array}{l}
(1+3\lambda)(1+\lambda^2)^3-(1+\lambda)^3
=\lambda^3(8+3\lambda+9\lambda^2+\lambda^3+3\lambda^4)\geq0,\\[0.8ex]
(1+3\lambda)^2-(1+\lambda)^3
=\lambda(3+6\lambda-\lambda^2)\geq0,\qquad
(1+\lambda)^3-(1+3\lambda)=\lambda^2(3+\lambda)\geq0.
\end{array}
\end{equation*}
Here the second inequality uses $0<\lambda\leq1$. This proves the lemma.
\end{proof}

\begin{proof}[\textbf{Proof of Lemma~\ref{lem: w=4 factor}}]
Write $B(\lambda):=1+4\lambda+6\lambda^2$. If $v=0$, then $\varPhi_{0,\Delta}(\lambda)\leq 1=B(\lambda)^0$. 
Assume $v\geq 1$. Since $v+\Delta\leq 4$, the possible pairs are
\begin{equation*}
(v,\Delta)\in\{(1,0),(1,1),(1,2),(1,3),(2,0),(2,1),(2,2),(3,0),(3,1),(4,0)\}.
\end{equation*}
A direct calculation gives
\begin{equation*}
\begin{array}{lll}
\varPhi_{1,0}=\varPhi_{1,1}=1,
& \displaystyle \varPhi_{1,2}=\frac{1+\lambda}{1+\lambda^2},
& \displaystyle \varPhi_{1,3}=\frac{1+\lambda}{1+\lambda^3},\\[1.2ex]
\varPhi_{2,0}=1+\lambda,
& \displaystyle \varPhi_{2,1}=\frac{1+3\lambda}{1+\lambda},
& \displaystyle \varPhi_{2,2}=\frac{(1+\lambda)^2}{1+\lambda^2},\\[1.2ex]
\varPhi_{3,0}=\varPhi_{3,1}=1+3\lambda,
& \varPhi_{4,0}=1+4\lambda+3\lambda^2. &
\end{array}
\end{equation*}
For $v=1$, it is enough to bound $\varPhi_{1,3}$, since
$\varPhi_{1,3}\geq \varPhi_{1,2}\geq \varPhi_{1,1}=\varPhi_{1,0}$. 
For $v=2$, it is enough to bound $\varPhi_{2,2}$, since
\begin{equation*}
\varPhi_{2,2}-\varPhi_{2,1}
=
\frac{2\lambda^2(1-\lambda)}{(1+\lambda)(1+\lambda^2)}
\geq0,
\qquad
\varPhi_{2,2}-\varPhi_{2,0}
=
\frac{\lambda(1-\lambda^2)}{1+\lambda^2}
\geq0.
\end{equation*}
Thus it remains to check
\begin{equation*}
\varPhi_{1,3}\leq B(\lambda)^{1/4},\qquad
\varPhi_{2,2}\leq B(\lambda)^{1/2},\qquad
\varPhi_{3,0}\leq B(\lambda)^{3/4},\qquad
\varPhi_{4,0}\leq B(\lambda).
\end{equation*}
After rearranging, the first three inequalities follow from
\begin{equation*}
\begin{array}{c}
\begin{aligned}
B(\lambda)(1+\lambda^3)^4-(1+\lambda)^4
&=\lambda^4\bigl(15+24\lambda+6\lambda^2+24\lambda^3+36\lambda^4+4\lambda^5\\
&\qquad\qquad\quad
+16\lambda^6+24\lambda^7+\lambda^8+4\lambda^9+6\lambda^{10}\bigr)\geq0,
\end{aligned}\\[1.2ex]
B(\lambda)(1+\lambda^2)^2-(1+\lambda)^4
=
2\lambda^2+4\lambda^3+12\lambda^4+4\lambda^5+6\lambda^6\geq0,\\[0.8ex]
B(\lambda)^3-(1+3\lambda)^4
=
12\lambda^2+100\lambda^3+315\lambda^4+432\lambda^5+216\lambda^6\geq0.
\end{array}
\end{equation*}
Finally, $\varPhi_{4,0}=1+4\lambda+3\lambda^2\leq B(\lambda)$. This proves the lemma.
\end{proof}

\end{document}